\documentclass[documentclass[runningheads]{llncs}
\usepackage{fancyhdr,graphicx,amsmath,amssymb}
\usepackage[ruled,vlined]{algorithm2e}
\usepackage{algorithmic}
\usepackage{float}
\usepackage{array}
\usepackage{multirow}
\usepackage{amssymb}
\usepackage{amsmath, mathtools}
\usepackage{graphicx}
\usepackage{calc}
\usepackage[section]{placeins}
\usepackage{subfig}
\usepackage{listings}
\usepackage{xcolor}
\usepackage{url}
\usepackage{enumitem}
\usepackage{textcomp}
\usepackage{scalerel}
\usepackage{booktabs}
\urldef{\mailsb}\path|jcarrasquel@hse.ru|
\urldef{\mailsa}\path|mecheraoui@misc-umc.org|
\urldef{\mailsc}\path|ilomazova@hse.ru|
\newcommand{\keywords}[1]{\par\addvspace\baselineskip
\noindent\keywordname\enspace\ignorespaces#1}
\usepackage{eurosym}
\def\code#1{\texttt{#1}}
\AtBeginDocument{}

\begin{document}

\mainmatter
\title{Compositional Conformance Checking of\\Nested Petri Nets and Event Logs of Multi-Agent Systems\thanks{This work is supported by the Basic Research Program at the National Research University Higher School of Economics.}}
\titlerunning{Compositional Conformance Checking for Nested Petri Nets and Event Logs}
\author{Khalil Mecheraoui\inst{1,2}, Julio C. Carrasquel\inst{1}, and Irina A. Lomazova\inst{1}}
\institute{$^{1}\;$National Research University Higher School of Economics,\\
Myasnitskaya ul. 20, 101000 Moscow, Russia\\ \mailsa , \mailsb , \mailsc \\
$^{2}\;$University of Constantine 2 --- Abdelhamid Mehri,\\Nouvelle ville Ali Mendjeli BP : 67A, 25000 Constantine, Algeria}
\maketitle
\vspace{-15pt}

\begin{abstract}
This paper presents a compositional conformance checking approach between nested Petri nets and event logs of multi-agent systems. By projecting an event log onto model components, one can perform conformance checking between each projected log and the corresponding component. We formally demonstrate the validity of our approach proving that, to check fitness of a nested Petri net is equivalent to check fitness of each of its components. Leveraging the multi-agent system structure of nested Petri nets, this approach may provide specific conformance diagnostics for each system component as well as to avoid to compute artificial boundaries when decomposing a model for conformance checking.
.
\end{abstract}

\vspace{-20pt}	
\section{Introduction}

Lift, thrust, drag, and gravity are the four forces helping an airplane fly. Process mining has similarly four forces to measure its quality namely fitness, generalization, precision, and simplicity \cite{vanDerAalst2016}. Conformance checking, which is one of the three pillars of process mining, is actually the fitness force \cite{Carmona2018}. It allows to check how well modeled behavior conforms reality as recorded in an event log. Conformance checking has become relevant in areas such as business alignment \cite{vanDerAalst2005}, auditing \cite{vanDerAalst2010}, and financial software testing \cite{Itkin2019}. However, current conformance checking approaches fall short when analyzing large event logs of complex multi-agent systems. These systems are characterized by a large number of agents interacting, and exhibiting a high degree of concurrency. In this light, it makes sense to use compositional approaches, where a conformance problem can be decomposed into smaller problems (e.g. \cite{vanDerAalst2013,MunozGama2014,Lee2018}). In \cite{vanDerAalst2013}, the author formalizes the so-called \textit{valid decomposition} to decompose conformance problems. This decomposition approach represents no problem from a conformance point of view. In \cite{MunozGama2014}, the authors proposed to decompose models using the idea of \emph{single-entry} and \emph{single-exit} (SESE). A SESE component is a subnet that has a simple interface w.r.t the rest of the net. Another approach is presented in \cite{Lee2018} to compute the overall conformance of a model by merging previously decomposed fragments.

Nevertheless, these approaches use process models whose structure may not provide a clear distinction between system components and their boundaries. This leads these approaches to increase their complexity to compute such boundaries (e.g., where to decompose a model, how many components, etc). Moreover, it may happen that the decomposition is \emph{artificial}, i.e., fragments of a decomposed model do not represent a real division of a system, so diagnostics for each real component may not be provided. In this sense, we propose the use of models of multi-agent systems. In particular, we consider nested Petri nets (NP-nets) \cite{Lomazova2000} --- an extension of Petri nets, where tokens can be Petri nets themselves, allowing to model multi-agent systems. NP-nets have been already used in the broader context of process modeling and workflow management \cite{vanHee2007,Lomazova2008}.

Fig. \ref{fig:model-npn} depicts an example of a NP-net describing an \emph{automated assistant engine} that can serve multiple customers concurrently. A NP-net consists of a \emph{system net}, i.e., modeling the system's environment, and a set of \emph{net tokens}, denoting interacting agents. Each net token has an inner Petri net structure describing agent behavior.

\vspace{-5pt}
\begin{figure}[H]
\centering
\includegraphics[width=9.20cm, height=7.13cm]{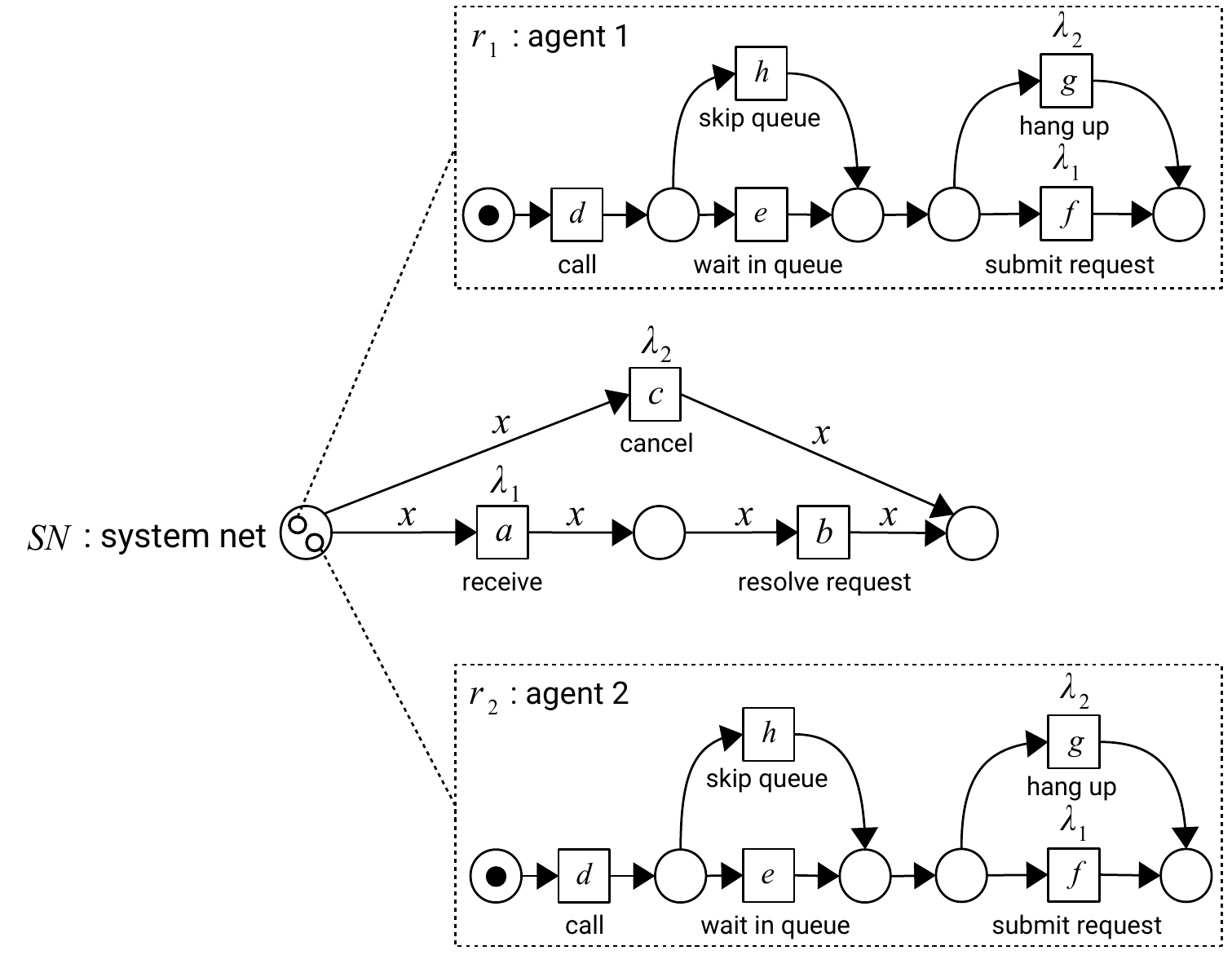}
\caption{A nested Petri net where the \emph{system net} $SN$ models an \emph{automated assistant engine}, serving customers concurrently (in this case, \emph{agents} $r_1$ and $r_2$).}
\label{fig:model-npn}
\end{figure}
\vspace{-15pt}

In this paper, we present a compositional conformance checking approach between nested Petri nets and event logs of multi-agent systems. Given an event log of a multi-agent system, we decompose it into several \textit{projections} according to the model components. Then, a conformance checking technique (e.g., replay, alignment) can be performed separately between each projection and the corresponding model component  (an agent or the system net). We assume that each agent in the event log corresponds to a net token in the nested Petri net model. For this task, we provide clear definitions regarding a subclass of nested Petri nets and event logs of multi-agent systems. To demonstrate the validity of our approach, we consider the notion of \emph{fitness}. If a model has \textit{perfect} fitness, then all log traces can be replayed on the model from beginning to end. In this work, we map such notion of an event log \textit{perfectly fits} a model, by defining how an event log of a multi-agent system fits a nested Petri net. Consequently, as an important result of this paper, we state and prove the following theorem: an event log of a multi-agent system \emph{perfectly fits} a nested Petri net if and only if the event log is syntactically correct w.r.t to the nested Petri net and each projection perfectly fits the corresponding model component. This theorem justifies the validity of our compositional approach.

The remainder of this paper is structured as follows. In section 2, we describe nested Petri nets. In section 3, we define the structure for event logs of multi-agent systems. In section 4, we present the compositional conformance checking approach of nested Petri nets and event logs of multi-agent systems, including the aforementioned theorem and its proof. Finally, section 5 presents some conclusions and future work.
\vspace{-5pt}

\section{Nested Petri Nets}

$\mathbb{N}$ denotes the set of natural numbers (including zero). Let $S$ be a set. The set of all subsets of $S$ is called a \emph{power set}, denoted as $\mathcal{P}(S)$, e.g., the power set of $S = \{a,b\}$ is $ \mathcal{P}(S) = \{ \{a,b\}, \{a\}, \{b\}, \emptyset \}$. A $multiset$ over $S$ is a mapping $m: S \to \mathbb{N}$. In other words, a multiset is a collection of elements, each of them with certain multiplicity, e.g., $\{b, a, b\},\;\{a, a, b\},$ and $\emptyset$ are multisets over $S$. For compactness, we write $\{{a}^{3}, {b}^{2}$\} for $\{a, a, a, b, b\}$. $\mathbb{P}_{m}(S)$ denotes the set of all multisets over $S$.
$\sigma = \langle a_1, a_2,..., a_n \rangle \in {S}^{*}$ denotes a \emph{sequence} of length $n$ over a set $S$.

\begin{definition}[Petri net]
A Petri net is a triple $N = (P, T, F)$, where $P$ is the set of places, $T$ is the set of transitions, $P \cap T = \emptyset$, and $F \subseteq (P \times T) \cup (T \times P)$ is the set of directed arcs (flow relation).
\end{definition}
\vspace{-4pt}

\noindent Petri nets \cite{Murata1989} is a formalism for modeling and analyzing concurrent distributed systems. As defined above, a Petri net consists of $places$ and $transitions$, which correspond respectively to $conditions$ and $activities$ of a system. Places may contain tokens, representing resources, control threads, etc. A $marking$ $m:P \to \mathbb{N}$ is a function that assigns tokens to places, denoting a system's state. The \emph{initial marking} is denoted as $m_{0}$, and the change into a new marking is defined by the \emph{firing rule}. Let $N = (P, T, F)$ be a Petri net, $X = P \cup T$, the sets $\prescript{\bullet}{}{x} = \{ y \in X | (y,x) \in F \}$ and $x^{\bullet} = \{ y \in X | (x,y) \in F \}$ denote the \textit{preset} and the \textit{postset} of $x \in X$. Transition $t \in T$ is \emph{enabled} in a \emph{marking} $m$ iff $\prescript{\bullet}{}{t} \subseteq m$. Then, the \emph{firing} of $t$ leads to a \textit{new marking} $m^{'} = m - \prescript{\bullet}{}{t} + t^{\bullet}$.

\begin{definition}[Workflow net]
Let $N = (P, T, F)$ be a Petri net. $N$ is a workflow net (WF-net) iff P contains a source place $i$ and a sink place $o$ s.t $\prescript{\bullet}{}{i} = {o}^{\bullet} = \emptyset $, and each node in $N$is on a path from $i$ to $o$.
\end{definition}
\vspace{-4pt}

\noindent When modeling individual agents in multi-agent systems, we consider \textit{workflow nets} \cite{vanDerAalst1998}. A WF-net has an initial and a final state, represented by markings $m_0 = \{i\}$ and $m_f = \{o\}$. Let $N = (P, T, F)$ be a WF-net, we consider an \emph{activity labeling} function $\delta : T \to A$, which assigns an \emph{activity} label to each transition $t \in T$, where $A$ is a finite set of activities. We define a sequence $\sigma = \langle a_1, ..., a_n \rangle \in A^{*}$ as a $run$ of a WF-net $N$ if there exists a firing sequence $\langle t_1, ...,t_n \rangle$ that leads from the initial marking $m_0 = \{i\}$ of $N$ to its final marking $m_f = \{o\}$ s.t $\delta(t_i) = a_i$ $(1 \leq i \leq n)$. The set of all possible runs of a WF-net $N$ is denoted by $\mathcal{B}(N)$ and is called the \emph{behavior} of $N$.

For modeling complex systems, one can use colored Petri nets (CP-nets). In CP-nets, tokens are attached with values belonging to different domains (\emph{color} types). Let $\mathcal{U}$ be the set of these different domains. Then, each place in a CP-net is typed with a domain in $\mathcal{U}$ indicating the type of tokens it contains. Arcs in CP-nets are annotated with expressions from a language $Expr$ defined over $Atom = V \cup C$, where $V$ and $C$ are sets of variables and constants. $Expr$ is defined as follows: $(i)$ An $atom \in Atom$ is an expression in $Expr$, $(ii)$ if $e_1$, $e_2$ $\in$ $Expr$, then $(e_1 + e_2)$ is an expression in $Expr$.

\begin{definition}[Colored Petri net]
A colored Petri net is a tuple $CPN = (P, T, F, \textup{\code{type}}, W )$ where:
\begin{itemize}
\item $(P,T,F)$ is a Petri net;
\item $\textup{\code{type}}: P \to \mathcal{U}$ is a place-typing function, mapping each place to a type in $\mathcal{U}$;
\item $W: F \to Expr$ is an arc expression function. $\; \forall r \in F$, if $r$ is adjacent to a place $p \in P$, then the type of $W(r)$ corresponds to the type of $p$.
\end{itemize}
\end{definition}

Let $CPN = (P, T, F, \textup{\code{type}}, W)$ be a CP-net over a set of domains $\mathcal{U}$. A marking $M$ in CPN is a function that maps each place $ p\in P $ into a multiset of tokens $ M(p)\in \mathbb{P}_{m}(\textup{\code{type}}(p))$. For a CPN, we distinguish an initial marking $M_0$ and a set of final markings $\Omega$. A binding $b$ of a transition $t$ is a function that assigns a value $b(v)$ to each variable $v$ occurring in the expression of an arc adjacent to $t$. For each variable $v$, $b(v)\in\bigcup_{Q \in \mathcal{U}}Q$. A pair $(t,b)$, where $b$ is a binding of $t$, is called a \emph{binding element}. An evaluation $W(p,t)(b)$ determines token demands (multiset of tokens) on $p$ for $t$ to be enabled with the binding $b$, and the multiset of tokens that the transition $t$ removes from the place $p$ when $t$ occurs with the binding $b$. $W(t,p)(b)$ determines the multiset of tokens added to an output place $p$. A transition is enabled in a marking $M$ w.r.t a binding $b$ iff for all $ p\in P $, $ W(p,t)(b)\subseteq M(p) $. An enabled transition fires in a marking $M$ yielding a new marking $M^{'}$, such that for all places $p$, $M^{'}(p) = (M(p) \setminus W(p,t)(b)) \cup W(t,p)(b)$.

In the following we consider a subclass of nested Petri nets (NP-nets). A NP-net $NP$ consists of a colored Petri net called the \emph{system net} $SN$, and a set of WF-nets $\mathcal{N} = \{E_1,..., E_n\}$ called \emph{element nets}, which define types of net tokens. In a system net $SN$, places contain either a set of \emph{net tokens} or a multiset of \emph{atomic colored tokens}. A net token is a marked element net, whereas an atomic colored token is a data value of some domain $D \in \mathcal{D}$, where $\mathcal{D}$ is a finite set of domains. Regarding the system net, we consider a language of expressions $Expr$ defined over $Atom = V \cup C$, where: (i) $V$ is a finite set of variables, typed over the set of element nets $\mathcal{N}$ and data domains $\mathcal{D}$ (e.g., the type of $x \in V$ is $E_1$) and (ii) $C$ is a finite set of constants, typed only over the set of data domains $\mathcal{D}$. Each arc $r$ is supplied with an arc expression from $Expr$. This arc expression can be either: a sum of variables typed over $\mathcal{N}$ if \code{type}$(p)$ $\in \mathcal{P}(\mathcal{N})$  where $p$ is a place adjacent to arc $r$ containing net tokens, or an arbitrary sum of distinct variables and constants typed over $\mathcal{D}$ if \code{type}$(p) \in \mathcal{D}$ where $p$ is a place adjacent to arc $r$ containing atomic colored tokens.

\begin{definition}[Nested Petri net] Let $\mathcal{D}$ be a finite set of domains, $Lab$ --- a finite set of synchronization labels and $A$ --- a finite set of activities. A nested Petri net (NP-net) is a tuple $NP = (SN, (E_1,...,E_k) ,\lambda, \delta)$, where:
\begin{itemize}
\vspace{-4pt}
\item $SN = (P_{SN}, T_{SN}, F_{SN}, \textup{\code{type}}, W)$ is a colored Petri net (called a system net) with two sets of places $P_{Net}$ and $P_{Atom}$ ($P_{Net} \cup P_{Atom} = P_{SN}$), such that for all $ p\in P_{Net}$, $\textup{\code{type}}(p)\subseteq \{E_1,...,E_k\}$ and for all $ p\in P_{Atom}$, $\textup{\code{type}}(p)\in \mathcal{D}$;
\vspace{1pt}
\item $for$ $i = \overline{1,k},$ $E_i = (P_i, T_i, F_i)$ is a WF-net, called an element net, s.t. $(T_{SN} \cup T_{i}) \cap (P_{SN} \cup P_{i}) = \emptyset$;
\vspace{1pt}
\item $\lambda: T_{NP} \nrightarrow Lab$ is a (partial) synchronization labeling function, where $T_{NP} = T_{SN} \cup T_{1} \cup ... \cup\; T_{k};$
\vspace{1pt}
\item $\delta : T_{NP} \to A$ is an activity labeling function, s.t $for$ $i = \overline{1,k}$ $T_{SN}\cap T_{i} = \emptyset$.
\end{itemize}
\end{definition}

In what follows, we consider \emph{conservative} NP-nets \cite{Lomazova2016}. In a conservative NP-net $NP = (SN, (E_1,...,E_k) ,\lambda, \delta)$, net tokens cannot be \emph{cloned} or \emph{disappear}. In a run, there is a \emph{stable set} of net tokens which we distinguish using individual \textit{agent names}. Let $R$ be a set of agent names, we propose a function $\code{class}: R \to \{E_1,...,E_k\}$, which maps to each agent name $r$ an element net. We denote by $(r, m)$ a net token which is characterized by an individual agent name $ r $ with the corresponding element net $ \code{class}(r) $ and a marking $ m $. The set of all possible net tokens is denoted by $S_{\scaleto{\code{agent}}{6.7pt}}$.

A marking $M$ in a NP-net $NP$ is a function mapping each place $p \in P_{SN}$ to a subset of $S_{\scaleto{\code{agent}}{6.7pt}}$ or a multiset over a domain $D \in \mathcal{D}$, in accordance with the type of $p$. Hence, elements in $M(p)$ are either distinguishable net tokens or atomic colored tokens which can be repeated. We say that a net token $(r, m)$ \emph{resides} in a place $p$ (under marking $M$) if $(r,m) \in M(p)$. Thus, the marking of a NP-net is defined by the marking of its system net. For a NP-net, we distinguish an initial marking $M_0$ and a set of final markings $M_f\in\Omega$.

Let $t$ be a transition in the system net $SN$ of a NP-net, and let be $\prescript{\bullet}{}{t} = \{p_{1},...,p_{i}\}$ and ${t}^{\bullet} = \{q_{1},...,q_{j}\}$ the sets of pre- and post-elements of transition $t$. $W(t) = \{W(p_1,t),...,W(p_i,t),W(t,q_1),W(t,q_j) \}$ denotes the set of all arc expressions adjacent to $t$. A \emph{binding} of $t$ is a function $b$ assigning to each variable $v$, occurring in each expression in $W(t)$, a value $b(v) \in$ $ \bigcup_{D \in \mathcal{D}} D\;\cup\; S_{\scaleto{\code{agent}}{6.7pt}}$. A transition $t$ in $SN$ is \emph{enabled} in a marking $M$ w.r.t a binding $b$ if for all $p \in \prescript{\bullet}{}{t} \; W(p,t)(b) \subseteq M(p)$. An enabled transition fires in a marking $M$ yielding a new marking $M^{'}$, such that for all places $p \in P_{SN}$, $M^{'}(p) = (M(p) \setminus W(p,t)(b)) \cup W(t,p)(b)$. For net tokens from $S_{\scaleto{\code{agent}}{6.7pt}}$ serving as variable values in input arc expressions from $W(t)$, we say that they are \emph{involved} in the firing of $t$. They are removed from input places and brought to output places of $t$.

\vspace{1.5pt}
\noindent We consider three kinds of steps in a NP-net:\\
\vspace{-5pt}

\noindent \textit{Element-autonomous step}: let $t$ be a transition without a synchronization label in a net token named $ r $, i.e., $\lambda(t)$ is not defined. When $ t $ is enabled in a marking $ m $, an \emph{element-autonomous step} is a firing of $t$ in marking $ m $, producing a new marking $ m' $, according to the usual firing rules of WF-nets.  This is also written as:\vspace{-2pt}
\[ m\xrightarrow{\delta(t), r} m' \]

\noindent \textit{System-autonomous step}: let $t \in T_{SN}$ be a transition without a synchronization label in the system net $SN$. A \emph{system-autonomous step} (also called a \textit{transfer step} when net tokens are involved) is the firing of transition $t$ according to the firing rule described above for a NP-net. The occurrence of this step in a marking $ M $ w.r.t a binding $b$, producing a new marking $ M' $, is denoted by:\vspace{-2pt}
\[M \xrightarrow{\delta(t), b} M' \]

\noindent \textit{Synchronization step}: let $t \in T_{SN}$ be a transition with a synchronization label $\lambda(t)$, and enabled in a marking $M$ w.r.t a binding $b$, and let $(r_{1},m_{1}),...,(r_{n},m_{n}) \in S_{\scaleto{\code{agent}}{6.7pt}}$ be net tokens involved in the firing of $t$. Then, $t$ can fire provided that in each $(r_{i},m_{i})$ ($1 \leq i \leq n$) there is an enabled transition labeled with the same value $\lambda(t)$. Thus, a \textit{synchronization step} goes in two stages: first, the firing of transitions $t_1,...,t_n$ in all net tokens involved in the firing of $t$, and then, the firing of $t$ in $SN$ w.r.t. a binding $b$. This step is denoted by: \vspace{-2pt}
\[ M\xrightarrow{\delta(t),\lbrace (\delta(t_{i}), r_{i}), 1\leq i\leq n\rbrace, b} M'\]
 
\begin{definition}[Run, Behavior of a nested Petri net]
\label{Run-Behavior}
Let $NP$ be a conservative nested Petri net and $\sigma$ --- a sequence of steps in $NP$. The occurrence of $\sigma$ from the initial marking $M_{0}$ of $ NP $, results in some final marking $ M_{f}\in\Omega $, is called a \textit{run}. The set of all possible runs is denoted by $\mathcal{B}(NP)$ and is called the behavior of $NP$.
\end{definition}

\section{Event Logs of Multi-Agent Systems}

An \emph{event log} of a multi-agent system is a multiset of traces, where a \emph{trace} is a sequence of events. Events consist of an \emph{activity} name, \emph{resources} which executed the activity or were \emph{involved} in its execution, and an (optional) multiset of data values. As possible resources we consider a system $SN$ or a finite set of agents with distinct names $r_1,r_2,...,r_n$. As shown in table \ref{tab:table1}, we consider three event types: (1) execution of an activity $a_1$ by some resource $r_1$, (2) execution of an activity $a$ by the system $SN$ where $n$ resources (agents) are involved, or (3) the simultaneous execution of activity $a$ by $SN$, and activities $a_{1},...,a_{n}$ by resources $r_{1},...,r_{n}$. For cases (2) and (3), events may contain $m$ \emph{data values} used by $SN$. We proceed to formally define a trace and an event log of a multi-agent system.

\begin{table}
\vspace{-20pt}
\caption{Event attributes in event logs of multi-agent systems.}
\label{tab:table1}
\centering
\begin{tabular}{|l|c|c|c|}
\hline
\multicolumn{1}{|c|}{\textbf{Event type}} & \multicolumn{3}{|c|}{\textbf{Attributes}}\\
\hline
\textit{event refers to...}& \textbf{Activity} & \textbf{Resource} & \textbf{Data} \\
\hline
(1) an \emph{element-autonomous} step & $a_1$ & $r_1$ & none. \\
\hline
(2) a \emph{system-autonomous} step & $a$ & $(SN, \{r_1,...,r_n\})$ & $\{d_1,...,d_m\}$ \\
\hline
(3) a \emph{synchronization} step & $(a, a_1,...,a_n)$ & $(SN, r_1,...,r_n )$ & $\{d_1,...,d_m\}$ \\
\hline
\end{tabular}
\vspace{-10pt}
\end{table}

\begin{definition}[Trace, Event log of a multi-agent system]
Let $ SN $ be a system name, $S$ --- a set of system activities, $ D $ --- a set of data, $B$ --- a set of agent activities, $R_{B}$ --- a set agent names, and $ \mathcal{C} $ --- a set where $C\in \mathcal{C}\Leftrightarrow C = C_{1}\cup C_{2}\, / \, C_{1}\subseteq R_{B}$ and $C_{2}\in \mathbb{P}_{m}(D)$. A \textit{trace} is a sequence $\sigma\in( \mathcal{E}_{{\scaleto{\code{agent}}{6.7pt}}}\cup\mathcal{E}_{{\scaleto{\code{system}}{6.7pt}}}\cup\mathcal{E}_{{\scaleto{\code{sync}}{5.4pt}}})^{*} $ where $ \mathcal{E}_{{\scaleto{\code{agent}}{6.7pt}}}= B\times R_{B}$, $ \mathcal{E}_{{\scaleto{\code{system}}{6.7pt}}}=  (S\times\lbrace SN\rbrace \times \mathcal{C})$, and $ \mathcal{E}_{{\scaleto{\code{sync}}{5.4pt}}}= (S\times\lbrace SN\rbrace \times \mathcal{P}(B\times R_{B})\times \mathbb{P}_{m}(D))$. $\mathcal{L}\in\mathbb{P}_{m}(\mathcal{E}_{{\scaleto{\code{agent}}{6.7pt}}},\mathcal{E}_{{\scaleto{\code{system}}{6.7pt}}},\mathcal{E}_{{\scaleto{\code{sync}}{5.4pt}}})^{*}$ is an event log, i.e., a multiset of traces.
\end{definition}

\begin{table}[h]
\vspace{-30pt}
  \begin{center}
    \caption{An event log of a multi-agent system $\mathcal{L}$ in tabular form, whose expected behavior is modeled in fig. \ref{fig:model-npn}.}
    \label{tab:table2}
     \begin{tabular}{|p{0.84\linewidth}|c|}
      \hline
      \textbf{Trace} & \textbf{Freq.} \\
      \hline
      $\sigma_1 = $ $\langle \; (d,r_{1}),(d,r_{2}), (h,r_{1}),(h,r_{2}), (a,SN,\{(f,r_{1})\}),(a,SN,\{(f,r_{2})\}),$\newline$(b,SN,\lbrace r_{2}\rbrace), (b,SN,\lbrace r_{1}\rbrace)\; \rangle$ & 4\\
      \hline
      $\sigma_2 = $ $ \langle \; (d,r_{2}),(h,r_{2}),(d,r_{1}), (h,r_{1})(a,SN,\{(f,r_{2})\}),(a,SN,\{(f,r_{1})\}),$\newline$(b,SN,\lbrace r_{1}\rbrace), (b,SN,\lbrace r_{2}\rbrace)\; \rangle $ & 1\\
      \hline
       $\sigma_3 = $ $ \langle \; (d,r_{2}),(e,r_{2}),(d,r_{1}),(h,r_{1}),((c,SN,\{(g,r_{1})\}),(c,SN,\{(g,r_{2})\})\; \rangle $ & 1\\
      \hline
       $\sigma_4 = $ $ \langle \; (d,r_{1}),(d,r_{2}),(h,r_{2}),(e,r_{1}),(c,SN,\{(g,r_{2})\}),(c,SN,\{(g,r_{1})\})\;  \rangle$ & 1\\
      \hline
       $\sigma_5 = $ $\langle \; (d,r_{1}),(d,r_{2}),(e,r_{1}),(e,r_{2}),(a,SN,\{(f,r_{1})\}),(c,SN,\{(g,r_{2})\}),$\newline$(b,SN,\lbrace r_{1}\rbrace)\; \rangle$ & 2\\
      \hline
    \end{tabular}
  \end{center}
\vspace{-20pt}
\end{table}

Table \ref{tab:table2} shows an event log $\mathcal{L}$ of the multi-agent system modeled in fig. \ref{fig:model-npn}. $\mathcal{L}$ contains information on nine traces. A distinct trace can occur multiple times. For instance, trace $\sigma_5$ occurred two times. It is a sequence of seven events. First, both activities $d$ and $e$ were executed by agents $r_1$ and $r_2$. Next, the system $SN$ executed two synchronization steps with agents $r_1$ and $r_2$, where activities $a$ and $f$, and later $c$ and $g$, were executed simultaneously. The trace ended by a system-autonomous step where $SN$ executed $b$ for agent $r_{1}$.

Let $\mathcal{L}$ be an event log of a multi-agent system and $NP$ --- a nested Petri net. $\mathcal{L}$ is \emph{syntactically correct} w.r.t. $NP$ if each event in $\mathcal{L}$ is \emph{syntactically correct} w.r.t. a step in $NP$ where:
\vspace{-2pt}

\begin{itemize}
\item An event $ (a,r) $ is syntactically correct w.r.t. a step in $NP$ if there is a transition $ t $ without synchronization label in a net token named $ r $ where $ \delta(t)=a $, and $ t $ can fire in a marking $ m $ producing a new marking $ m' $, i.e., $ m\xrightarrow{\delta(t), r} m'$ in $ NP $ where $ t\in E_{i}, 1\leq i\leq k,$ and $ \code{class}(E_{i})=r$.

\item An event $ (a,SN,{d_{1},...,d_{m}}) $ is syntactically correct w.r.t. a step in $NP$ if there is a transition $ t\in T_{SN} $ without synchronization label where $ \delta(t)=a $, and $ t $ can fire in a marking $ M $ w.r.t. a binding $ b $ such that $ b $ assigns the atomic tokens $ d_{1},...,d_{m} $ to the variables in $ W(t) $, i.e., $ M\xrightarrow{\delta(t), b} M'$ in $ NP $.

\item An event $ (a,SN,\{(a_{1},r_{1}),...,(a_{q},r_{q})\},\lbrace d_{1},...,d_{p}\rbrace) $ is syntactically correct w.r.t. a step in $NP$ if there is a transition $ t\in T_{SN} $ with synchronization label $ \lambda(t) $ where $ \delta(t)=a $, $ (r_{1},m_{1}),...,(r_{n},m_{n}) $ are net tokens involved in the firing of $ t $ such that in each $ (r_{i},m_{i}) $  there is an enabled transition $ t_{i}$ ($\delta(t_{i})=a_{i}$)  labeled with the same value $ \lambda(t) $, and $ t $ can fire in a marking $ M $ w.r.t. a binding $ b $ assigning the atomic tokens $d_{1},...,d_{p} $ and also the agent names $ r_{1},...,r_{q} $ to the variables in $ W(t) $, i.e., $ M\xrightarrow{\delta(t),\lbrace (\delta(t_{i}), r_{i}), 1\leq i\leq n\rbrace, b} M' $ in $ NP $.\vspace{-2pt}
\end{itemize}

\section{Compositional Conformance Checking of Nested Petri Nets and Event Logs of Multi-Agent Systems}

In this section, we propose a solution to check conformance of event logs of multi-agent systems and nested Petri nets. We prove that an event log perfectly fits a NP-net iff the event log is syntactically correct w.r.t. that NP-net, and each projection of the event log onto a NP-net component perfectly fits that component.

\noindent Let $\langle\rangle$ denote an empty sequence, $\sigma_{1}.\sigma_{2}$ --- concatenation of two sequences, and $\sigma {\upharpoonright}_{e}$ --- the $projection$ of sequence $\sigma$ on an element $e$.

\emergencystretch 5em%
\begin{definition}[Trace projection onto an agent]
\label{Projection1} 
 Let $X = (\mathcal{E}_{{\scaleto{\code{agent}}{6.7pt}}}\cup\mathcal{E}_{{\scaleto{\code{system}}{6.7pt}}}\cup\mathcal{E}_{{\scaleto{\code{sync}}{5.4pt}}})$ be a set of events, $B$ --- a set of agent activities, and $R_{B}$ --- a set of agent names. ${\upharpoonright}_{r\in R_{B}}\in X^{*}\longrightarrow B^{*}$ is a projection function defined recursively: $ (1) \langle\rangle{\upharpoonright}_{r\in R_{B}}=\langle\rangle$ $ (2) $ For $ \sigma\in X^{*}$ and $ e\in X $:

\begin{itemize}[nosep]
\item  If $e = (x,y)\in \mathcal{E}_{{\scaleto{\code{agent}}{6.7pt}}},y=r$ then $(\langle e \rangle.\sigma){\upharpoonright}_{r\in R_{B}}=\langle x \rangle.\sigma{\upharpoonright}_{r\in R_{B}} $;
 \item   If $e=(x,SN,\{(x_{1},y_{1}),...,(x_{n},y_{n})\},\lbrace d_{1},...,d_{p}\rbrace)\in \mathcal{E}_{{\scaleto{\code{sync}}{5.4pt}}}, y_{i}=r, 1\leq i\leq n$ then $(\langle e \rangle.\sigma){\upharpoonright}_{r\in R_{B}}=\langle x_{i} \rangle.\sigma{\upharpoonright}_{r\in R_{B}} $;
 \item  Otherwise $ (\langle e \rangle.\sigma){\upharpoonright}_{r\in R_{B}}=  \sigma{\upharpoonright}_{r\in R_{B}} \quad\quad\quad \quad$.
\end{itemize}
\end{definition}

\emergencystretch 5em%
\begin{definition}[Trace projection onto a system net]
\label{Projection2} 
 Let $X = (\mathcal{E}_{{\scaleto{\code{agent}}{6.7pt}}}\cup\mathcal{E}_{{\scaleto{\code{system}}{6.7pt}}}\cup\mathcal{E}_{{\scaleto{\code{sync}}{5.4pt}}})$ and $Y\subseteq S\times \mathcal{C} $ where $ S $ is a set of system activities, $C\in \mathcal{C}\Leftrightarrow C = C_{1}\cup C_{2} / C_{1}\subseteq R_{B}$ and  $C_{2}\in \mathbb{P}_{m}(D)$, $ R_{B} $ is a set of agent names, and $ D $ is a set of data. ${\upharpoonright}_{SN}\in X^{*}\longrightarrow Y^{*}$ is a projection function defined recursively:  (1) $\langle\rangle{\upharpoonright}_{SN}=\langle\rangle$ (2) For $ \sigma\in X^{*}$ and $ e\in X $:

\begin{itemize}[nosep]
 \item   If $e =(x,SN,\lbrace d_{1},...,d_{p}\rbrace)\in \mathcal{E}_{{\scaleto{\code{system}}{6.7pt}}}$ then $(\langle e \rangle.\sigma){\upharpoonright}_{SN}=\langle(x,\lbrace d_{1},...,d_{p}\rbrace) \rangle.\sigma{\upharpoonright}_{SN} $;
 \item   If $e=(x,SN,\{(x_{1},y_{1}),...,(x_{n},y_{n})\},\lbrace d_{1},...,d_{p}\rbrace)\in \mathcal{E}_{{\scaleto{\code{sync}}{5.4pt}}}$ then $(\langle e \rangle.\sigma){\upharpoonright}_{SN}=\langle (x,\lbrace y_{1},...,y_{n},d_{1},...,d_{p}\rbrace) \rangle.\sigma{\upharpoonright}_{SN} $;
 \item  Otherwise $ (\langle e \rangle.\sigma){\upharpoonright}_{SN}=  \sigma{\upharpoonright}_{SN} \quad\quad\quad \quad$.
\end{itemize}

\end{definition}

\noindent Def. \ref{Projection1} (resp. Def. \ref{Projection2}) is used to project traces of an event log onto net tokens (resp. a system net). A projection of a trace onto a net token yields the sequence of agent activities. A projection onto the system net yields a sequence of pairs where each pair consists of an activity and the set of resources and data involved in this activity execution. For instance, consider the projection of the event log $\mathcal{L}$ (cf. Table \ref{tab:table2}) onto components of the NP-net $N$ (cf. Fig. \ref{fig:model-npn}). Table \ref{tab:table3} show the three decomposed event logs $L_{SN}$, $L_{1}$, and $L_{2}$ resulting from the projection of $\mathcal{L}$ onto (a) the system net, (b) agent $r_1$, and (c) agent $r_2$ respectively.

\begin{table}
\vspace{-15pt}
\caption{Projections $L_{SN}$, $L_{1}$, and $L_{2}$ from the event log $\mathcal{L}$ (cf. Table \ref{tab:table2}) onto (a) the system net $SN$, and agents (b) $r_1$ and (c) $r_2$ from $N$ (cf. Fig. \ref{fig:model-npn}).}
\label{tab:table3}
\vspace{3pt}
\subfloat[caption][$L_{SN}$]{\begin{tabular}{ | l | c |} 
\hline
\textbf{Trace} & \textbf{Freq.}\\ 
\hline
$\langle(a,\lbrace r_{1}\rbrace),(a,\lbrace r_{2}\rbrace),(b,\lbrace r_{2}\rbrace),(b,\lbrace r_{1}\rbrace)\rangle$ & 4\\
\hline
$\langle(a,\lbrace r_{2}\rbrace),(a,\lbrace r_{1}\rbrace),(b,\lbrace r_{1}\rbrace),(b,\lbrace r_{2}\rbrace)\rangle$ & 1\\
\hline
$\langle(c,\lbrace r_{1}\rbrace),(c,\lbrace r_{2}\rbrace)\rangle $ & 1 \\
\hline
$ \langle(c,\lbrace r_{2}\rbrace),(c,\lbrace r_{1}\rbrace)\rangle$ & 1 \\ 
\hline
$ \langle(a,\lbrace r_{1}\rbrace),(c,\lbrace r_{2}\rbrace),(b,\lbrace r_{1}\rbrace)\rangle $ & 2 \\
\hline
\end{tabular}}
\hspace{1.33em}
\subfloat[caption][$L_{1}$]{\begin{tabular}{ | c | c |} 
\hline
\textbf{Trace} & \textbf{Freq.}\\ 
\hline
$\langle d, h, f\rangle$ & 5\\
\hline
$\langle d, h, g \rangle$ & 1\\
\hline
$\langle d, e, g \rangle$ & 1\\
\hline
$\langle d, e, f \rangle$ & 2\\
\hline
\end{tabular}}
\hspace{1.33em}
\subfloat[caption][$L_{2}$]{\begin{tabular}{ | c | c | c | c |} 
\hline
\textbf{Trace} & \textbf{Freq.}\\ 
\hline
$\langle d, h, f\rangle$ & 5\\
\hline
$\langle d, e, g \rangle$ & 3\\
\hline
$\langle d, h, g\rangle$ & 1\\
\hline
\end{tabular}}
\vspace{-22pt}
\end{table}
 
\noindent Thus, a conformance checking technique can be applied to each projection and the corresponding NP-net component, ignoring their synchronization labels. In particular, for the system net, we replace net tokens by their agent names, which are atomic colored tokens. We consider synchronization steps as autonomous steps, and for a marking $M$ in a NP-net $NP$, marking projections onto $NP$ components are defined as follows: (1) the projection of $M$ onto a system net $SN$, denoted as $M{\upharpoonright}_{SN}$, is a marking of the colored Petri net $SN$ obtained by replacing all net tokens in $M$ by their agent names, and (2) the projection of $M$ onto a net token $(r,m)$, denoted as $ M{\upharpoonright}_{(r,m)}$, is just $m$. A system net component $SN$, with a marking $ M{\upharpoonright}_{SN} $ and without synchronization labels, is a CP-net labeled by activity names. A sequence of binding elements $\langle (t_{1},b_{1}),...,(t_{n},b_{n})\rangle$, starting from the initial marking $ M_{0}{\upharpoonright}_{SN} $ and ending in a final marking $ M_{f}{\upharpoonright}_{SN} $, projected onto the set of system net activities is called a $ \textit{run} $.

\begin{definition}[Perfectly fitting event log]
\label{Perf-Fits}
Let $\mathcal{L}$ be an event log of a multi-agent system and $NP$ --- a nested Petri net. $\mathcal{L}$ perfectly fits $NP$ if and only if for all $\sigma=\langle e_{1},...,e_{n}\rangle\in\mathcal{L}$ there is a run $\sigma'=\langle s_{1},...,s_{n}\rangle\in\mathcal{B}(NP)$ such that for $ i = \overline{1,n} $, $ e_{i} $ is syntactically correct w.r.t $ s_{i} $.
\end{definition}

\noindent Let $NP$ be a nested Petri net, $\mathcal{L}$ --- an event log, $ (r_{1},m_{1}),...,(r_{n},m_{n}) $ --- net tokens of $ NP $, $ L_{1},...,L_{n} $ –- corresponding projections of $ \mathcal{L} $, and $ L_{SN} $ –- a projection of $ \mathcal{L} $ onto actions of the system net. $L_{SN}$ perfectly fits $ SN $ if and only if for all $\sigma=\langle e_{1},...,e_{m}\rangle\in L_{SN}$, there is a run $ \sigma' $ in the system net component where
\[\sigma'= (M_{0} \xrightarrow{\delta(t_{1}),b_{1}} M_{1}\xrightarrow{\delta(t_{2}),b_{2}} M_{2}...M_{m-1}\xrightarrow{\delta(t_{m}),b_{m}} M_{f})\] and for $ i = \overline{1,m} $, $ e_{i}=(a_{i},\{d^{i}_{1},...,d^{i}_{p}\}) $, $ \delta(t_{i})=a_{i} $, and $ b_{i} $ is the binding assigns $\{d^{i}_{1},...,d^{i}_{p}\}$ to the variables in $ W(t_{i}) $. For $i = \overline{1,n}$, $L_{i}$ perfectly fits $(r_{i},m_{i})$ if and only if for all $\sigma=\langle e_{1},...,e_{m}\rangle\in L_{i}$, there is a run $ \sigma' $ in the element net $ \code{class}(r_{i}) $ where
\[\sigma'= (m_{0} \xrightarrow{\delta(t_{1})} m_{1}\xrightarrow{\delta(t_{2})} m_{2}...m_{m-1}\xrightarrow{\delta(t_{m})} m_{f})\] and for $ j = \overline{1,m} $, $ e_{j}=a_{j} $ and $ \delta(t_{j})=a_{j} $.

An event log perfectly fits a model if all traces in the log can be replayed on the model from beginning to end. For instance, let us consider the event log $\mathcal{L}$ (cf. Table \ref{tab:table2}) and the NP-net $N$ depicted in fig. \ref{fig:model-npn}.  Clearly, $\mathcal{L}$ perfectly fits $NP$. Also, each projected event log $L_{SN}$, $L_{1}$, or $L_{2}$ (cf. Table \ref{tab:table3}) perfectly fits the corresponding component in $NP$.

\begin{theorem}
\label{fitness}
Given a nested Petri net $NP = (SN, (E_1,...,E_k) ,\lambda, \delta)$ and an event log $\mathcal{L}\in\mathbb{P}_{m}(\mathcal{E}_{{\scaleto{\code{agent}}{6.7pt}}},\mathcal{E}_{{\scaleto{\code{system}}{6.7pt}}},\mathcal{E}_{{\scaleto{\code{sync}}{5.4pt}}})^{*}$, let $ (r_{1},m_{1}),...,(r_{n},m_{n}) $ be net tokens of $ NP $, $ L_{1},...,L_{n} $ – corresponding projections of $ \mathcal{L} $, and $ L_{SN} $ – a projection of $ \mathcal{L} $ onto the system net $ SN $. $ \mathcal{L} $ perfectly fits $ NP $ if and only if:
\begin{enumerate}
\item 	$ \mathcal{L} $ is syntactically correct w.r.t $ NP $;
\item 	$ L_{SN} $ perfectly fits $ SN $;
\item 	$ L_{i} $ perfectly fits $(r_{i},m_{i})$, $ 1\leqslant i\leqslant n $. 
\end{enumerate}

\end{theorem}

\begin{proof}
Let $\mathcal{L}$ be an event log of a multi-agent system, $NP = (SN, (E_1,...,E_k) ,\lambda, \delta)$ --- a nested Petri net, and $ (r_{1},m_{1}),...,(r_{n},m_{n}) $ --- net tokens of $ NP $.

($ \Rightarrow $) Let $\sigma=\langle e_{1},...,e_{m}\rangle \in \mathcal{L} $ be such that there is a run $\sigma'=\langle s_{1},...,s_{m}\rangle \in \mathcal{B}(NP) $ and for $ j = \overline{1,m} $:
\begin{itemize}
\item if $ e_{i}=(a,r)$ then $ s_{i}= m\xrightarrow{\delta(t), r} m'$ where $ t\in E_{i}, 1\leq i\leq k, \code{class}(E_{i})=r,$ and $ \delta(t)=a $. \textbf{(1)}
\item if $ e_{i}=(a,SN,{d_{1},...,d_{p}})$ then $ s_{i}=M\xrightarrow{\delta(t), b} M'$ where $ \delta(t)=a,$ and $ b $ is a binding assigning the atomic tokens $ d_{1},...,d_{p} $ to the variables in $ W(t) $. \textbf{(2)}
\item if $ e_{i}=(a,SN,\{(a_{1},r_{1}),...,(a_{q},r_{q})\},\lbrace d_{1},...,d_{p}\rbrace)$ then $ s_{i}=M\xrightarrow{\delta(t),\lbrace (\delta(t_{i}), r_{i}), 1\leq i\leq q\rbrace, b} M'$ where $ \delta(t)=a,$, $ \delta(t_{i})=a_{i}, 1\leq i\leq q$, and $ b $ is a binding assigning the atomic tokens $d_{1},...,d_{p} $ and also the agent names $ r_{1},...,r_{q} $ to the variables in $ W(t) $. \textbf{(3)}
\end{itemize}
i.e., $ \sigma $ perfectly fits $ NP $. We need to prove that
\begin{itemize}
\item $ \sigma $ is syntactically correct w.r.t $ NP $;
\item $\sigma \upharpoonright_{SN}=\langle e^{s}_{1},...,e^{s}_{m'}\rangle$ perfectly fits the system net component, i.e., there is a run $ \sigma_{SN} $ in the system net component where:
\[\sigma_{SN}= (M_{0} \xrightarrow{\delta(t_{1}),b_{1}} M_{1}\xrightarrow{\delta(t_{2}),b_{2}} M_{2}...M_{m'-1}\xrightarrow{\delta(t_{m'}),b_{m'}} M_{f})\], and for $ i = \overline{1,m'} $, $ e^{s}_{i}=(a_{i},\{d^{i}_{1},...,d^{i}_{p}\}) $, $ \delta(t_{i})=a_{i} $, and $ b_{i} $ is the binding that assigns $\{d^{i}_{1},...,d^{i}_{p}\}$ to the variables in $ W(t_{i}) $;
\item for $ i = \overline{1,n}$, $\sigma \upharpoonright_{r_{i}}=\langle a^{i}_{1},...,a^{i}_{m''}\rangle$ perfectly fits $ (r_{i},m_{i}) $, i.e., there is a run $ \sigma_{r_{i}} $ in the element net $ \code{class}(r_{i}) $ where:
\[\sigma_{r_{i}}= (m_{0} \xrightarrow{\delta(t_{1})} m_{1}\xrightarrow{\delta(t_{2})} m_{2}...m_{m-1}\xrightarrow{\delta(t_{m''})} m_{f})\] and for $ j = \overline{1,m''}, \delta(t_{j})=a^{i}_{j}$.
\end{itemize}

By the fact that $ \sigma $ perfectly fits $ NP $, it follows trivially that $ \sigma $ is syntactically correct w.r.t $ NP $ (cf. Def. \ref{Perf-Fits}).

Taking into account that $\sigma'=\langle s_{1},...,s_{m}\rangle$ is a run in $ NP $ (which can hold synchronization labels) where for $ i = \overline{1,m} $ we have (1), (2) and (3), and that 
  \[
    \langle e_{i} \rangle\upharpoonright_{SN}=\left\{
                \begin{array}{ll}
                  \langle\rangle,\qquad\qquad\qquad\qquad\qquad\:\:\: if\: e_{i}=(a,r) \\
                  \langle (a,\{d_{1},...,d_{p}\})\rangle,\qquad\qquad\: if\: e_{i}=(a,SN,{d_{1},...,d_{p}})\qquad\quad\textbf{(4)}\\
                  \langle (a,\{r_{1},...,r_{q},d_{1},...,d_{p}\})\rangle,\: if\: e_{i}=(a,SN,\{(a_{1},r_{1}),...,\\ \qquad\qquad\qquad\qquad\qquad\qquad\qquad\quad (a_{q},r_{q})\},d_{1},...,d_{p}\rbrace)
                \end{array}
              \right.
  \]
(cf. Def. \ref{Projection2}) we deduce that for $\sigma \upharpoonright_{SN}=\langle e^{s}_{1},...,e^{s}_{m'}\rangle$, there is a run $ \sigma_{SN} $ in the system net component where:
\[\sigma_{SN}= (M_{0} \xrightarrow{\delta(t_{1}),b_{1}} M_{1}\xrightarrow{\delta(t_{2}),b_{2}} M_{2}...M_{m'-1}\xrightarrow{\delta(t_{m'}),b_{m'}} M_{f})\]
and for $ i = \overline{1,m'} $, $ e^{s}_{i}=(a_{i},\{d^{i}_{1},...,d^{i}_{p}\}) $, $ \delta(t_{i})=a_{i} $, and $ b_{i} $ is the binding that assigns $\{d^{i}_{1},...,d^{i}_{p}\}$ to the variables in $ W(t_{i}) $. Therefore, $\sigma \upharpoonright_{SN}$ perfectly fits the system net component.

Now by the fact that $\sigma'$ is a run in $ NP $ (which can hold synchronization labels) where for $ i = \overline{1,m} $ we have (1), (2) and (3), and that for $ j = \overline{1,n} $
  \[
    \langle e_{i} \rangle\upharpoonright_{r_{j}}=\left\{
                \begin{array}{ll}
                  \langle a \rangle,\qquad\qquad\quad\: if\: e_{i}=(a,r_{j})\\
                  \langle a_{j}\rangle, 1\leq j\leq q,\: if\: e_{i}=(a,SN,\{(a_{1},r_{1}),...,(a_{q},r_{q})\},\qquad\quad\textbf{(5)}\\ \qquad\qquad\:\:\:\qquad\qquad\qquad\qquad\quad\lbrace d_{1},...,d_{p}\rbrace)\\
                  \langle\rangle,\qquad\qquad\quad\:\:\: otherwise
                \end{array}
              \right.
  \]
(cf. Def. \ref{Projection1}) it follows that for $\sigma \upharpoonright_{r_{i}}=\langle a_{1},...,a_{m''}\rangle, 1\leq i\leq n$, there is a run $ \sigma_{r_{i}} $ in the element net $ \code{class}(r_{i}) $ where:
\[\sigma_{r_{i}}= (m_{0} \xrightarrow{\delta(t_{1})} m_{1}\xrightarrow{\delta(t_{2})} m_{2}...m_{m-1}\xrightarrow{\delta(t_{m''})} m_{f})\] and for $ j = \overline{1,m''}, \delta(t_{j})=a_{j}$. Therefore, $\sigma \upharpoonright_{r_{i}}$ perfectly fits $ (r_{i},m_{i}) $.

($ \Leftarrow $) Let $\sigma=\langle e_{1},...,e_{m}\rangle \in \mathcal{L} $ be such that $ \sigma $ is syntactically correct w.r.t $ NP $, $\sigma \upharpoonright_{SN}$ perfectly fits the system net component, and for $ i = \overline{1,n} $, $\sigma \upharpoonright_{r_{i}}$ perfectly fits $ (r_{i},m_{i}) $. We need to prove that $ \sigma $ perfectly fits $ NP $.

Taking into account (4) and (5), and that $ \sigma $ is syntactically correct w.r.t $ NP $, we deduce that by associating to each element of the projected sequences the corresponding resource, elements can be merged together into the trace $ \sigma $. Therefore, being $\sigma \upharpoonright_{SN}$ perfectly fits the system net component and  for $ i = \overline{1,n} $ $\sigma \upharpoonright_{r_{i}}$ perfectly fits $ (r_{i},m_{i}) $, then $ \sigma $ perfectly fits $ NP $ and this achieves the proof. 
\end{proof}

\section{Conclusions and Future Work}

In this paper, we proposed a compositional approach for conformance checking of nested Petri nets and event logs of multi-agent systems. Nested Petri nets are a well-known Petri net extension where tokens can be Petri nets themselves, allowing to model multi-agent systems. An event log can be projected onto NP-net components (system net and all agents), so conformance checking can be performed between each projection and the corresponding component. This approach can provide specific conformance diagnostics for each system component. We demonstrated the validity of our approach proving that, an event log perfectly fits a nested Petri net if and only if it is syntactically correct w.r.t the model and each projection perfectly fits the corresponding model component. For future research, we consider the experimental evaluation of our approach against other approaches when checking conformance of multi-agent systems.

\vspace{-4pt}
\bibliography{references}

\end{document}